\pdfoutput=1
\newif\ifFull
\Fulltrue
\ifFull
\documentclass[11pt]{article}
\advance \topmargin by -\headheight
\advance \topmargin by -\headsep
\evensidemargin \oddsidemargin
\else
\documentclass{sig-alternate}
\fi
\usepackage{graphicx}
\usepackage{url}
\usepackage{times}
\usepackage[noend]{algorithmic}

\makeatletter
\def\@begintheorem#1#2{\sl \trivlist \item[\hskip \labelsep{\bf #1\ #2:}]}
\def\@opargbegintheorem#1#2#3{\sl \trivlist
      \item[\hskip \labelsep{\bf #1\ #2\ #3:}]}
\makeatother

\newtheorem{theorem}{Theorem}
\newtheorem{lemma}[theorem]{Lemma}

\newtheorem{corollary}[theorem]{Corollary}

\ifFull
\newenvironment{proof}{\noindent{\bf Proof:}}{\hspace*{\fill}\rule{6pt}{6pt}\bigskip}
\else
\renewcommand{\paragraph}[1]{\smallskip\noindent\textbf{#1}\hspace{1ex}}
\fi

\begin{document}
\ifFull\else
\conferenceinfo{SPAA'11,} {June 4--6, 2011, San Jose, California, USA.}
\CopyrightYear{2011}
\crdata{978-1-4503-0743-7/11/06}
\clubpenalty=10000
\widowpenalty = 10000
\hyphenpenalty=500
\tolerance=700
\setlength{\emergencystretch}{1ex}
\fi

\title{Data-Oblivious External-Memory Algorithms for the \\
Compaction, Selection, and Sorting of Outsourced Data}

\author{
{Michael T. Goodrich} \\
Dept.~of Computer Science \\ 
University of California, Irvine \\
\url{http://www.ics.uci.edu/~goodrich/}
}

\date{}

\maketitle 

\begin{abstract}

We present data-oblivious
algorithms in the external-memory model for 
compaction, selection, and sorting.
Motivation for such problems comes 
from clients
who use outsourced data storage services and wish 
to mask their data access patterns.
We show that compaction and selection
can be done data-obliviously using $O(N/B)$ I/Os, and sorting can be done,
with a high probability of success,
using $O((N/B)\log_{M/B} (N/B))$ I/Os.
\ifFull
Our methods use a number of new algorithmic techniques, including 
data-oblivious uses of invertible Bloom lookup tables, a butterfly-like
compression network, randomized data thinning,
and ``shuffle-and-deal'' data perturbation.
In addition, since data-oblivious sorting is the bottleneck
in the ``inner loop'' in existing
oblivious RAM simulations, our sorting result improves
the amortized time overhead to do oblivious RAM simulation 
by a logarithmic factor in the external-memory model.
\fi
\end{abstract}

\ifFull\else
\vspace{-0.1in}
\vspace{-0.1in}
\vspace{-0.1in}
\fi

\section{Introduction}

Online data storage is becoming a large and growing industry, in which
companies provide large storage capabilities for hosting outsourced 
data for individual or corporate users, such as with
Amazon S3 and
Google Docs.
\ifFull
As an example of the size and growth of this industry, 
we note that in March 2010 the Amazon S3
service passed the milestone of hosting over 100 billion 
objects~\footnote{%
  \url{www.datacenterknowledge.com/archives/2010/03/09/amazon-s3-now-hosts-100-billion-objects/}},
which was double the number from the year before.

\fi
The companies that provide such services typically give their users
guarantees about the availability and integrity of their data, but
they often have commercial interests in learning as much as possible
about their users' data.
\ifFull
For instance, it may be in such a company's commercial
interest to search a user's documents for keywords that
could trigger advertisements that are better 
targeted towards a user's interests.
Of course, a significant portion of a user's data might
be sensitive, either for personal or proprietary reasons; hence, there
is a need for techniques for keeping a user's data private in such
circumstances.
\fi

Clearly, a necessary first step towards achieving privacy
is for a user to store his or her data in encrypted form, e.g., using
a secret key that is known only to the user.
Simply encrypting one's data is not sufficient to
achieve privacy in this context, however,
since information about the content of a user's
data is leaked by the pattern in which the user accesses it.
For instance, Chen {\it et al.}~\cite{cwwz-sclwa-10} show that
highly sensitive information
can be inferred from user access patterns at popular
health and financial web sites even if the contents of those communications
are encrypted.

\paragraph{Problem Statement.}
The problem of protecting the privacy of a user's data 
accesses in an outsourced data storage facility
can be defined in terms of 
\emph{external-memory data-oblivious RAM computations}.
In this framework, a user, Alice, 
has a CPU with a private cache, which can be used for data-dependent
computations.
She stores the bulk of her data on 
a data storage server owned by Bob, who provides 
her with an interface that supports indexed addressing of her data,
so each word has a unique address.
As in the standard external-memory model (e.g., see~\cite{av-iocsr-88}), 
where
external storage is typically assumed to be on a disk drive, 
data on the external storage device
is accessed and organized in contiguous blocks,
with each block holding $B$ words, where $B\ge 1$.
That is, this is an adaptation of the standard
external-memory model so that
the CPU and cache are associated with Alice and the external disk
drive is associated with Bob.
The service provider, Bob, is trying to learn as much as possible about the
content of Alice's data and he can view the sequence and location of all
of Alice's disk accesses. 
But he cannot see the content of what is read or
written, for we assume it is encrypted using a semantically secure
encryption scheme such that re-encryption of the same value is
indistinguishable from an encryption of a different value.
Moreover, Bob cannot view the content or access patterns
of Alice's private cache, the size of which, $M$, is allowed to be function
of the input size, $N$, and/or, the block size, $B$.
Using terminology of the cryptography literature, Bob is assumed to
be an \emph{honest-but-curious} adversary~\cite{gmw-hpamg-87}, 
in that he correctly
performs all protocols and does not tamper with any data, but he is
nevertheless interested in learning as much as possible about
Alice's data.

We say that Alice's sequence of I/Os is \emph{data-oblivious} 
for solving some problem, $\cal P$ (like sorting or selection), if the
distribution of this sequence depends only on $\cal P$, $N$, $M$, and $B$,
and the length of the access sequence itself. 
In particular, this distribution should be independent of the
data values in the input.
Put another way, this definition of a data-oblivious computation 
means that $\Pr(S\, |\, {\cal M,P},N,M,B)$,
the probability that Bob sees an access sequence, $S$,
conditioned on a specific initial configuration of external memory, $\cal M$, 
and the values of $\cal P$, $N$, $M$, and $B$,
satisfies
\[
\Pr(S\, |\, {\cal M},{\cal P},N,M,B) = \Pr(S\, |\, {\cal M}',{\cal P},N,M,B),
\]
for any memory configuration
${\cal M}' \not= {\cal M}$ such that $|{\cal M}'|=|{\cal M}|$.
Note, in particular, that this implies that the length, $|S|$, of
Alice's access sequence cannot depend on specific input values.
Examples of data-oblivious access sequences for an array, $A$,
of size $N$, in Bob's external memory, 
include the following (assuming $B=1$):
\begin{itemize}
\item
Simulating a circuit, $\cal C$, with its inputs taken in order from $A$.
For instance, $\cal C$ could be a Boolean circuit or an AKS sorting
\ifFull
network~\cite{aks-osn-83,aks-scps-83}.
\else
network~\cite{aks-osn-83}.
\fi
\item
Accessing the cells of $A$ according to a random hash function, $h(i)$,
as $A[h(1)]$, $A[h(2)]$, $\ldots$, $A[h(n)]$, or random permutation,
$\pi(i)$,
as $A[\pi(1)]$, $A[\pi(2)]$, $\ldots$, $A[\pi(n)]$.
\end{itemize}
Examples of computations on $A$ that would \emph{not} be data-oblivious 
include the following:
\begin{itemize}
\item
Using the standard merge-sort or quick-sort algorithm to sort $A$.
(Neither of these well-known algorithms is data-oblivious.)
\item
Using values in $A$ as indices for a hash table, $T$,
and accessing them
as $T[h(A[1])]$, $T[h(A[2])]$, $\ldots$, $T[h(A[n])]$, where $h$ is a
random hash function.
(For instance, think of what happens if all the values in $A$ are
equal and how improbable the resulting $n$-way collision
in $T$ would be.)
\end{itemize}
This last example is actually a little subtle. 
It is, in general, not data-oblivious to access a hash table, $T$, in
this way, but note that such an access
sequence \emph{would} be data-oblivious if the elements in $A$ 
are always guaranteed to be distinct.
In this case, each such
access in $T$ would be to a random location and every possible sequence
of such accesses in $T$ would be equally likely for any set of
distinct values in $A$, assuming our random hash function, $h$, satisfies
the random oracle model 
(e.g., see~\cite{br-roap-93}).

\paragraph{Related Prior Results.}
Data-oblivious sorting is a classic algorithmic
problem, which Knuth~\cite{k-ss-73} studies in some depth, since
deterministic schemes give rise to sorting networks, such as the
theoretically-optimal
\ifFull
$O(n\log n)$ AKS network~\cite{aks-osn-83,aks-scps-83,p-isnod-90,s-snlgf-09}
as well as practical sorting networks~\cite{l-ipaaa-92,p-ssn-72}.
\else
$O(n\log n)$ AKS network~\cite{aks-osn-83}
as well as practical sorting networks~\cite{l-ipaaa-92}.
\fi
Randomized data-oblivious sorting algorithms running in
$O(n\log n)$  time and succeeding with high probability are
likewise studied by Leighton and
Plaxton~\cite{lp-hsn-98} and Goodrich~\cite{g-rsaso-10}.
Moreover, data-oblivious sorting is finding applications to
privacy-preserving secure multi-party
computations~\cite{wlgdz-bpstp-10}.
For the related selection problem,
Leighton {\it et al.}~\cite{lms-pnsms-97} give a randomized
data-oblivious solution that runs in $O(n\log\log n)$ time and they
give a matching lower bound
for methods exclusively based on compare-exchange
operations.
\ifFull
In addition,
we note that the data-oblivious compaction problem is strongly
related to the construction of $(n,m)$-concentrator networks~\cite{c-ocsgn-79}.
\fi

Chaudhry and Cormen~\cite{cc-scnoa-06} 
argue that data-oblivious external-memory 
sorting algorithms have several advantages, including 
good load-balancing and the avoidance
of poor performance on ``bad'' inputs.
They give a data-oblivious sorting algorithm for the parallel disk
model and analyze its 
performance experimentally~\cite{DBLP:conf/esa/ChaudhryC05}. But their
method is size-limited and
does not achieve the optimal I/O complexity of previous 
external-memory sorting algorithms 
(e.g., see~\cite{av-iocsr-88}), 
which use $O((N/B)\log_{M/B} (N/B))$ I/Os.
None of the previous
external-memory sorting algorithms that use 
$O((N/B)\log_{M/B} (N/B))$ I/Os
are data-oblivious, however.

Goodrich and Mitzenmacher~\cite{gm-mpcho-10}
show that any
RAM algorithm, $\cal A$, can be simulated in a data-oblivious 
fashion in the external-memory 
model so that each 
memory access performed by $\cal A$ has an overhead of
$O(\min\{\log^2 (N/B),\,\log^2_{M/B} (N/B)\log_2 N\})$ 
amortized I/Os in the simulation, and
which is data-oblivious with
high probability (w.v.h.p.)\footnote{In this paper, we take
  the phrase ``with high probability'' to mean that the
  probability is at least $1-1/(N/B)^d$, for a given constant $d\ge 1$.},
which both improves and extends a result of 
Goldreich and Ostrovsky~\cite{go-spsor-96} to the I/O model.
The overhead of the Goodrich-Mitzenmacher simulation is optimal for the case
when $N/B$ is polynomial in $M/B$ and 
is based on an ``inner loop'' use of a deterministic 
data-oblivious sorting algorithm that uses
$O((N/B)\log^2_{M/B} (N/B))$ I/Os, but this result is not optimal in
general.

\paragraph{Our Results.}
We give data-oblivious external-memory algorithms for
compaction, selection, quantile computation, and sorting of
outsourced data.
Our compaction and selection algorithms use $O(N/B)$ I/Os and
our sorting algorithm uses $O((N/B)\log_{M/B} (N/B))$ I/Os,
which is asymptotically optimal~\cite{av-iocsr-88}.
All our algorithms are randomized 
and succeed with high probability.
Our results also imply
that one can improve the expected 
amortized overhead for a randomized external-memory
data-oblivious RAM simulation to be $O(\log_{M/B} (N/B)\log_2 N)$.

Our main result is for the sorting problem, as we are not aware of
any asymptotically-optimal oblivious external-memory sorting
algorithm prior to our work.
Still, our other results are also important,
in that our sorting algorithm is based on our 
algorithms for selection and quantiles,
which in turn are based on 
new methods for data-oblivious data compaction.
We show, for instance, that efficient compaction leads to a
selection algorithm that uses $O(n)$ I/Os, 
and succeed w.v.h.p, where $n=O(N/B)$.
Interestingly, this result beats a lower bound for the complexity
of parallel selection networks, 
of $\Omega(n\log\log n)$, due to Leighton {\it et al.}~\cite{lms-pnsms-97}.
Our selection result doesn't invalidate their lower bound, however, for 
their lower bound is for circuits that only use
compare-exchange operations, whereas our algorithm uses other
``blackbox''
primitive operations besides compare-exchange, including 
addition, subtraction, data copying, and random functions.
Thus, our result
demonstrates the power of using operations other
than compare-exchange in a data-oblivious selection algorithm.

Our algorithms are based on 
a number of new techniques for data-oblivious algorithm design, including
the use of a recent data structure
by Goodrich and Mitzenmacher~\cite{goodrich2011invertible}, 
known as the invertible
Bloom lookup table.
We show in this paper how this data structure
can be used in a data-oblivious way to solve the
compaction problem, which in turn leads to efficient methods for
selection, quantiles, and sorting.
Our methods also use a butterfly-like routing network and a
``shuffle-and-deal'' technique reminiscent of 
Valiant-Brebner routing~\cite{vb-uspc-81},
so as to avoid data-revealing access patterns.

Given the motivation of our algorithms from outsourced data privacy,
we make some reasonable assumptions regarding
the computational models used by some of our algorithms.
For instance, for some of our results, 
we assume that $B\ge\log^\epsilon (N/B)$, which we call the
\emph{wide-block} 
assumption. This is, in fact, equivalent or weaker than several
similar assumptions in the external-memory literature 
\ifFull
(e.g., see~\cite{bdf-cobt-05,bf-lbemd-03,mm-embfs-02,ps-siod-09,rs-opasp-08,%
wyz-dehlb-09,yz-ocpcd-10}).
\else
(e.g., see~\cite{bdf-cobt-05,mm-embfs-02}).
\fi
In addition, we sometimes also use
a weak \emph{tall-cache} 
assumption that $M\ge B^{1+\epsilon}$,
for some small constant $\epsilon>0$,
which is also common in the external-memory literature
\ifFull
(e.g., see~\cite{DBLP:conf/swat/Brodal04,bf-olco-03,cacheoblivious05,%
bdf-cobt-05,DBLP:conf/icalp/BrodalFM05,flpr-coa-99}).
\else
(e.g., see~\cite{DBLP:conf/swat/Brodal04,flpr-coa-99}).
\fi
\ifFull
For example, if we take $\epsilon=1/2$ and Alice wishes to sort an
array of up to $2^{67}$ 64-bit words (that is, a zettabyte of data), then, 
to use our algorithms, she would need a block size of 
at least $8$ words (or 64 bytes) and a private cache of at least 23
words (or 184 bytes).
In addition,
we assume throughout 
that keys and values can be stored in memory words or blocks of
memory words, which support the operations of 
read, write, copy, compare, add, and subtract, as in the standard RAM model.
\fi

\section{Invertible Bloom Filters}
As mentioned above,
one of the tools we use in our algorithms involves a data-oblivious use of
the invertible Bloom lookup table of 
Goodrich and Mitzenmacher~\cite{goodrich2011invertible},
which is itself based on
the invertible Bloom filter data structure of 
Eppstein and Goodrich~\cite{eg-sirrt-10}.

An invertible Bloom lookup
table, ${\cal B}$, is a randomized data structure storing
a set of key-value pairs.
It supports the following operations (among others):
\begin{itemize}
\item
insert$(x,y)$: insert the key-value pair, $(x,y)$, into ${\cal B}$.
This operation always succeeds, assuming that keys are distinct.
\item
delete$(x,y)$: remove the key-value pair, $(x,y)$, from ${\cal B}$.
This operation assumes $(x,y)$ is in ${\cal B}$.
\item
get$(x)$: lookup and return the value, $y$, associated with the 
key, $x$, in $\cal B$.
This operation may fail, with some probability.
\item
listEntries: list all the key-value pairs being stored in ${\cal B}$.
With low probability, this operation
may return a partial list along with an ``list-incomplete'' error condition.
\end{itemize}

When an invertible-map Bloom lookup table ${\cal B}$ is 
first created, it initializes
a table, $T$, of a specified capacity, $m$, which is initially empty.
The table $T$ is the main storage used to implement $\cal B$.
Each of the cells in $T$ stores a constant number of fields, each of
which is a single memory word or block.
Insertions and deletions can proceed independent of the capacity $m$ and
can even create situations where the number, $n$,
of key-value pairs in ${\cal B}$ can be much larger than $m$.
Nevertheless, the space used for $B$ remains $O(m)$.
The get and listEntries methods, on the other hand, only
guarantee good probabilistic success when $n<m$.

Like a traditional Bloom filter~\cite{b-sttoh-70}, an invertible Bloom
lookup table uses a set of $k$ random hash functions,
$h_1$, $h_2$, $\ldots$, $h_k$, defined on the universe of keys,
to determine where items are stored.
In this case, we also assume that, for any $x$, the $h_i(x)$ values
are distinct, which can be achieved by a number of methods, including
partitioning.
Each cell contains three fields:
\begin{itemize}
\item
a \texttt{count} field,
which counts the number of entries that have been mapped to this cell,
\item
a \texttt{keySum} field,
which is the sum of all the keys that have been mapped to this cell,
\item
a \texttt{valueSum} field,
which is the sum of all the values that have been mapped to this cell.
\end{itemize}

Given these fields, which are all initially set to $0$, performing 
the insert operation is fairly straightforward:
\begin{itemize}
\item
insert$(x,y)$:
\begin{algorithmic}[100]
\FOR {each (distinct) $h_i(x)$ value, for $i=1,\ldots,k$}
\STATE
add $1$ to $T[h_i(x)].\mbox{\texttt{count}}$
\STATE
add $x$ to $T[h_i(x)].\mbox{\texttt{keySum}}$
\STATE
add $y$ to $T[h_i(x)].\mbox{\texttt{valueSum}}$
\ENDFOR
\end{algorithmic}
\end{itemize}

The delete method is basically the reverse of that above.
The listEntries method is similarly simple\footnote{
	We describe this method in a destructive
	fashion---if one wants a non-destructive method, then one should first
	create a copy of the table $T$ as a backup.}:
\begin{itemize}
\item listEntries:
\begin{algorithmic}[100]
\WHILE {there is an $i\in[1,m]$ s.~t.~$T[i].\mbox{\texttt{count}} = 1$}
\STATE output
$(T[i].\mbox{\texttt{keySum}}\, ,\, T[i].\mbox{\texttt{valueSum}})$
\STATE call delete($T[i].\mbox{\texttt{keySum}}$)
\ENDWHILE
\end{algorithmic}
\end{itemize}
It is a fairly straightforward exercise to implement this method in $O(m)$
time, say, by using a link-list-based priority queue of cells in $T$ indexed by
their \texttt{count} fields and modifying the delete method to update this
queue each time it deletes an entry from $\cal B$.
If, at the end of the while-loop, all the entries in $T$ are empty, then we say
that the method succeeded.

\begin{lemma}[(Goodrich and Mitzenmacher)~\cite{goodrich2011invertible}]
\label{lem:bloom}
Given an invertible Bloom lookup table, $T$,
of size $m=\lceil \delta kn \rceil$, holding at most $n$ key-value pairs,
the listEntries method succeeds with probability $1-1/n^c$, 
where $c\ge 1$ is any given constant and $k\ge 2$ and $\delta\ge 2$
are constant depending on $c$.
\end{lemma}

The important observation about the functioning of the invertible
Bloom lookup table, for the purposes of this paper, is that 
the sequence of memory locations accessed during 
an insert$(x,y)$ method is oblivious to 
the value $y$ and the number of items already stored in the table.
That is, the locations accessed in performing an
insert method depend only on the key, $x$.
The listEntries method, on the other hand, is not oblivious to keys or values.

\section{Data-Oblivious Compaction}
In this section, we describe efficient data-oblivious
algorithms for compaction.
In this problem, we are given an array, $A$, of $N$ cells, 
at most $R$ of
which are marked as ``distinguished''
(e.g., using a marked bit or a simple test)
and we want to produce an array, $D$, of size $O(R)$ that contains all
the distinguished items from $A$.
Note that this is the fundamental operation done during disk
\ifFull
defragmentation (e.g., see~\cite{r-ansfp-69}), which is a natural operation
\else
defragmentation, which is a natural operation
\fi
that one would want to do in an outsourced file system, since users
of such systems are charged for the space they use.

We say that such a compaction algorithm is \emph{order-preserving}
if the distinguished items in $A$ remain 
in their same relative order in $D$.
In addition, we say that a compaction algorithm is \emph{tight}
if it compacts $A$ to an array $D$ of size exactly $R$.
If this size of $D$ is merely $O(R)$,
and we allow for
some empty cells in $D$, then we say that the compaction is \emph{loose}.
Of course, we can always use a data-oblivious sorting algorithm, such as 
with the
following sub-optimal result, to perform a tight order-preserving compaction.

\begin{lemma}[(Goodrich and Mitzenmacher)~\cite{gm-mpcho-10}]
\label{lem:sub-sort}
Given an array $A$ of size $N$, one can 
sort $A$ with a deterministic data-oblivious algorithm that uses
$O((N/B)\log^2_{M/B} (N/B))$ I/Os, assuming $B\ge 1$ and $M\ge 2B$.
\end{lemma}

In the remainder of this section, we give several compaction algorithms,
which exhibit various trade-offs
between performance and the compaction properties listed above. 
Incidentally, each of these algorithms 
is used in at least one of our algorithms for selection, quantiles,
and sorting.

\paragraph{Data Consolidation.}
Each of our compaction algorithms uses a data-oblivious consolidation 
preprocessing operation.
In this operation, we are given an array $A$ of size $N$ such that at most
$R\le N$ elements in $A$ are marked as ``distinguished.''
The output of this step is an array, $A'$, of $\lceil N/B\rceil$ blocks, such
that $\lfloor R/B\rfloor$ blocks in $A'$ are completely full of
distinguished elements and at most one block in $A'$ is partially
full of distinguished elements. All other blocks in $A'$ are
completely empty of distinguished elements.
This consolidation step uses $O(N/B)$ I/Os, 
assuming only that $B\ge 1$ and $M\ge 2B$, and it is order-preserving
with respect to the distinguished elements in $A$.

We start by viewing $A$ as being subdivided 
into $\lceil N/B\rceil$ blocks of size $B$ each
(it is probably already stored this way in Bob's memory).
We then scan the blocks of $A$ from beginning to end, keeping a block,
$x$, in Alice's memory as we go.
Initially, we just read in the first block of $A$ and let $x$ be this
block.
Each time after this that we scan a block, $y$, of $A$,
from Bob's memory,
if the distinguished 
elements from $x$ and $y$ in Alice's memory can form a full block, then we write
out to Bob's memory 
a full block to $A'$ of the first $B$ distinguished elements
in $x\cup y$, maintaining the relative order of these distinguished
elements.
Otherwise, we merge the fewer than $B$ distinguished elements in
Alice's memory into the single block, $x$, again, 
maintaining their relative order,
and we write out an empty block to $A'$ in Bob's memory.
We then repeat.
When we are done scanning $A$, we write out $x$ to Bob's memory; 
hence, the access
pattern for this method is a simple scan of $A$ and $A'$, which is clearly
data-oblivious. Thus,
we have the following.

\begin{lemma}
\label{lem:consolidate}
Suppose we are given an array $A$ of size $N$, with 
at most $R$ of its elements marked as distinguished.
We can deterministically consolidate $A$ into 
an array $A'$ of $\lceil N/B\rceil$ 
blocks of size $B$ each with a data-oblivious
algorithm that uses $\lceil N/B\rceil$ I/Os, such that all but possibly
the last block in $A'$ are 
completely full of distinguished elements or completely empty of 
distinguished elements.  This computation assumes 
that $B\ge 1$ and $M\ge 2B$, and it preserves the relative order of
distinguished elements.
\end{lemma}

\paragraph{Tight Order-Preserving Compaction for Sparse Arrays.}
The semi-oblivious property of the invertible Bloom lookup table
allows us to efficiently perform tight order-preserving 
compaction in a data-oblivious fashion, provided the input array is
sufficiently sparse.
Given the consolidation preprocessing step 
(Lemma~\ref{lem:consolidate}), we describe
this result in the RAM model (which could be applied to blocks that
are viewed as memory words for the external-memory model),
with $n=N/B$ and $r=R/B$.

\begin{theorem}
\label{thm:tight-sparse}
Suppose we are given an array, $A$,
of size $n$, holding at most $r\le n$ distinguished items. 
We can perform a tight order-preserving compaction of
the distinguished items from $A$ into an array $D$ of size $r$
in a data-oblivious fashion in the RAM model
in $O(n+r\log^2 r)$ time. 
This method succeeds with probability $1-1/r^c$, for any given
constant $c\ge1$.
\end{theorem}
\begin{proof}
We create an invertible Bloom lookup table, $T$, of size $3r$.
Then we map each entry, $A[i]$, into $T$ using the key-value pair $(i,A[i])$.
Note that the insertion algorithm begins by reading $k$ cells of $T$
whose locations in this case depend only on $i$ (recall that $k$ is
the number of hash functions).
If the entry $A[i]$ is distinguished, then this operation
involves changing the fields of these
cells in $T$ according to the invertible Bloom
lookup table insertion method and then writing them back to $T$.
If, on the other hand, $A[i]$ is not distinguished,
then we return the fields back to each cell $T[h_j(i)]$ unchanged
(but re-encrypted so that Bob cannot tell of the cells were changed
or not).
Since the memory accesses in $T$ are the same
independent of whether $A[i]$ is distinguished or not, each key we use
is merely an index $i$, and insertion into an invertible Bloom
lookup table is oblivious to all factors other than the key, this
insertion algorithm is data-oblivious.
Given the output from this insertion phase, which is the table $T$,
of size $O(r)$, we then perform a RAM
simulation of the listEntries method to list out the distinguished
items in $T$, using the simulation
result of Goodrich and Mitzenmacher~\cite{gm-mpcho-10}.
The simulation algorithm performs the actions of a RAM computation
in a data-oblivious fashion such that each step of the RAM
computation has an amortized time overhead of $O(\log^2 r)$.
Thus, simulating the listEntries method in a data-oblivious fashion
takes $O(r\log^2 r)$ time and succeeds with
probability $1-1/r^c$, for any given constant $c>0$,
by Lemma~\ref{lem:bloom}.
The size of the output array, $D$, is exactly $r$.
Note that, at this point, the items in the 
array $D$ are not necessarily in the same relative order that they
were in $A$.
So, to make this compaction operation order-preserving, we complete
the algorithm by performing a data-oblivious sorting of $D$, using
each item's original position in $A$ as its key, say by 
Lemma~\ref{lem:sub-sort}.
\end{proof}

Thus, we can perform tight data-oblivious compaction for sparse arrays,
where $r$ is $O(n/\log^2 n)$, in
linear time.  Performing this method on an array that is not sparse,
however, would result in an $O(n\log^2 n)$ running time.
Nevertheless, we can perform such an action on a dense array
faster than this.

\paragraph{Tight Order-Preserving Compaction.}
Let us now
show how to perform a tight order-preserving compaction for a dense
array in a data-oblivious fashion using $O((N/B)\log_{M/B} (N/B))$ I/Os.
We begin by performing a data consolidation preprocessing operation
(Lemma~\ref{lem:consolidate}).
This allows us to describe our compaction algorithm 
at the block level,
for an array, $A$, of $n$ blocks, where $n=\lceil N/B\rceil$,
and a private memory of size $m=O(M/B)$.

Let us define a routing network, $\cal R$, for performing this
action, in such a way that $\cal R$ is somewhat like a butterfly
network (e.g., see~\cite{l-ipaaa-92}).
There are $\lceil \log n\rceil$ levels, $L_0,L_1,\ldots$,
to this network, with each
level, $L_i$, consisting of $n$ cells (corresponding to the positions in $A$).
Cell $j$ of level $L_i$ is connected to cell $j$ and cell $j-2^i$
of level $L_{i+1}$.
(See Figure~\ref{fig:levels}.)

\begin{figure}[hbt!]
\ifFull
\vspace*{-0.3in}
\centering \hspace*{1.2in}\includegraphics[width=7in]{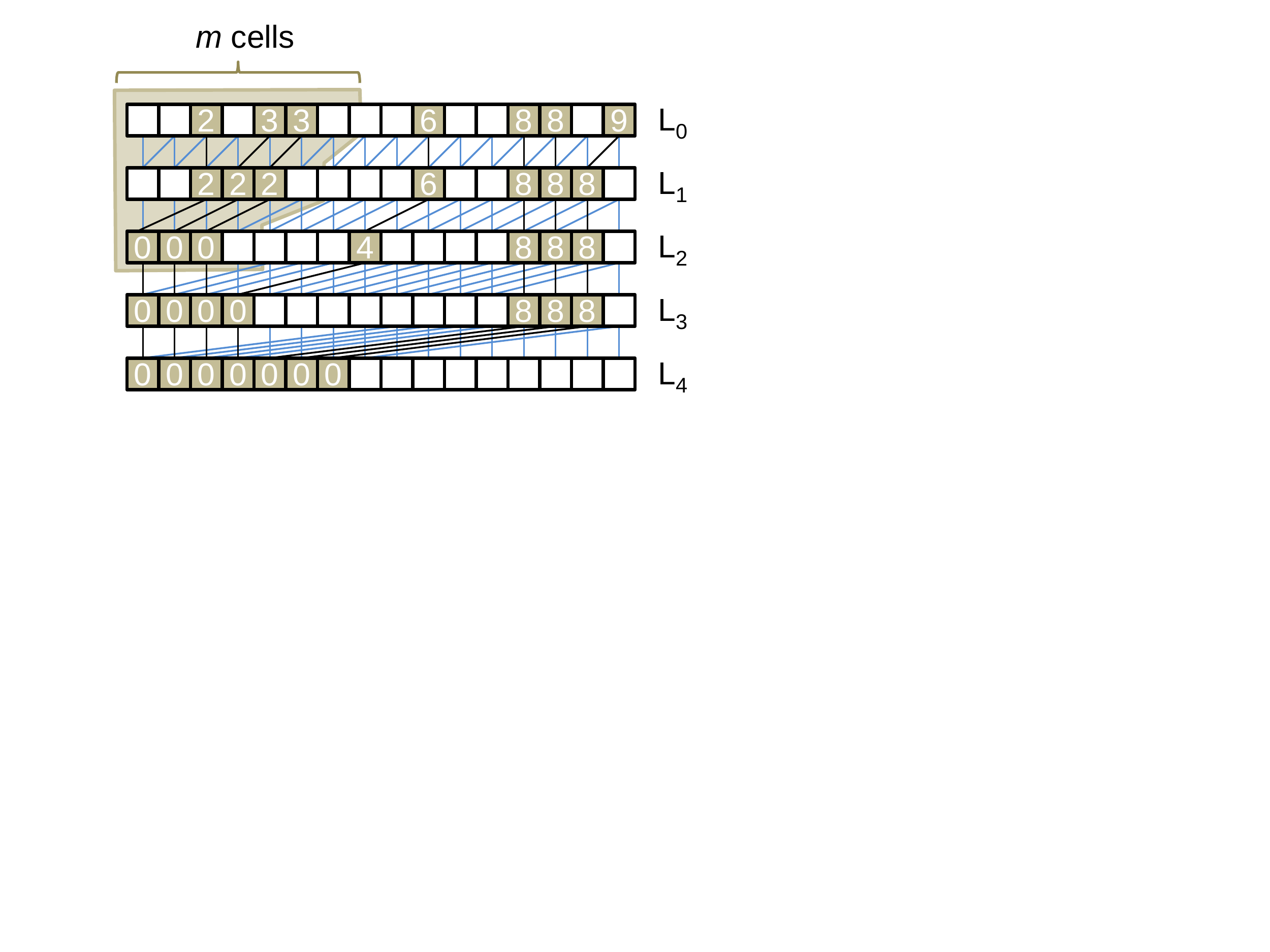}
\vspace*{-3.4in}
\else
\centering \includegraphics[width=3in, trim=0.6in 4.2in 4.5in 0.15in, clip]{levels}
\fi
\caption{\label{fig:levels} A butterfly-like compaction network.
Occupied cells are shaded and labeled with the remaining distance
that the block for that cell needs to move to the left. In addition,
a block of $m$ cells on level $L_0$ is highlighted to show that its
destination on a level $O(\log m)$ away is of size roughly $m/2$.}
\end{figure}

Initially, each occupied cell on level $L_0$ is labeled with the
number of cells that it needs to be moved to the left to create a
tight compaction. Note that such a labeling can easily be produced by
a single left-to-right data-oblivious scan of the array $A$.
For each occupied cell $j$ on level $L_i$, labeled with distance
label, $d_j$, we route the contents of cell $j$ to 
cell $j-(d_j \bmod 2^{i+1})$,
which, by a simple inductive argument,
is either cell $j$ or $j-2^i$ on level $L_{i+1}$.
We then update the distance label for this cell to be 
$d_j\leftarrow d_j-(d_j \bmod 2^{i+1})$, and continue.

\begin{lemma}
If we start with a valid set of distance labels on level
$L_0$, there will be no collisions at any internal level of the
network $\cal R$.
\end{lemma}
\begin{proof} 
Notice that we can cause a collision between two cells, $j$ and $k$,
in going from level $i$ to level $i+1$,
only if $k=j+2^i$ and $(d_j \bmod 2^{i+1})=0$ and $(d_k\bmod 2^{i+1})= 2^i$.
So suppose such a pair of colliding cells exists.
Then $(d_k-d_j) \bmod 2^{i+1} = 2^i$.
Also, note that there are $d_k-d_j$ empty cells between $k$ and $j$,
by the definition of these distance labels.
Thus, there are at least $2^i$ empty cells between $j$ and
$k$; hence, $k\ge j+2^i+1$, which is a contradiction.
So no such pair of colliding cells exists.
\end{proof}

Thus, we can route the elements of $A$ to their
final destinations in $O(n\log n)$ scans. We can, in fact, route
elements faster than this, by considering cells $2m-1$ at a time,
starting at a cell $jm+1$ on level $L_0$. In
this case, there are $m$ possible destination cells at
level $L_l$ with index at least $jm+1$,
where $l=i+\log m \, -\, 1$. 
So we can route all these $m$ cells with destinations among
these cells on level $L_l$ in internal memory 
with a single scan of these cells.
We can then move this ``window'' of $2m-1$ cells to the right by $m$
cells and repeat this routing. At the level $L_l$ we note that these
consecutive $m$ cells are now independent of one other, and the
connection pattern of the previous $\log m\, -\, 1$ levels will be
repeated for cells that are at distance $m$ apart.
Thus, we can repeat the same routing for level $L_l$ (after a simple
shuffle that brings together cells that are $m$ apart).
Therefore, we can route the cells from $L_0$ to level $L_{\log n}$ using
$O(n\log n/\log m) = O((N/B)\log_{M/B} (N/B))$ I/Os.
Moreover, since we are simulating a circuit, the sequence of I/Os is
data-oblivious.

\begin{theorem}
\label{thm:network}
Given an array $A$ of $N$ cells such that at most $R\le N$ 
of the cells in $A$ are
distinguished, we can deterministically perform a tight
order-preserving compaction of $A$ 
in a data-oblivious fashion using $O((N/B)\log_{M/B} (N/B))$ I/Os,
assuming $B\ge1$ and $M\ge 3B$.
\end{theorem}

Note that we can also use this method ``in reverse'' to expand any
compact array to a larger array in an order-preserving way.
In this case, each element would be given an expansion factor, which
would be the number of cells it should be moved to the right (with
these factors forming a non-decreasing sequence).

\paragraph{Loose Compaction.}
Suppose we are now given an array $A$ of size $N$, such that at most
$R<N/4$ of the elements in $A$ are marked as ``distinguished'' and
we want to map the distinguished elements from $A$ into an array $D$,
of size $5R$ using an algorithm that is data-oblivious.

We begin by performing the data consolidation algorithm of
Lemma~\ref{lem:consolidate}, to consolidate the distinguished elements
of $A$ to be full or empty block-sized cells (save the last one),
where we consider a cell ``empty'' if it stores a null value that is
different from any input value.
So let us view $A$ as a set of $n=\lceil N/B\rceil$ cells and
let us create an array, $C$, of size $4r$ block-sized cells,
for $r=\lceil R/B\rceil$.

Define an \emph{$A$-to-$C$ thinning pass} to be a sequential scan of
$A$ such that, for each $A[i]$,
we choose an index $j\in [1,4r]$ uniformly at random, read
the cell $C[j]$, and if it's empty, $A[i]$ is distinguished, and we
have yet to successfully write $A[i]$ to $C$ (which can be indicated
with a simple bit associated with $A[i]$), then we write
$A[i]$ to the cell $C[j]$.
Otherwise, if $C[j]$ is nonempty or $A[i]$ is not distinguished, then we write
the old value of $C[j]$ back to the cell $C[j]$.
Thus, the memory accesses made
by an $A$-to-$C$ thinning pass are data-oblivious and the time to
perform such a pass is $O(n)$.

We continue our algorithm by
performing $c_0$ rounds of $A$-to-$C$ thinning passes, where $c_0$ is
a constant determined in the analysis, so
the probability that any block of $A$ is
unsuccessfully copied into $C$ is at most $1/2^{2c_0}$.
We then consider regions of $A$ of size $\lceil c_1\log n\rceil=O(\log (N/B))$,
where $c_1$ is another constant determined in the analysis, so that 
each such region has at most $(c_1\log n)/2$ occupied cells w.v.h.p.
In particular, we assign these values according to the following.

\begin{lemma}
\label{lem:log-compress}
Consider a region of $c_1\log (N/B)$ blocks of $A$, such that each block
is unsuccessfully copied into $C$ independently 
with probability at most $1/2^{2c_0}$.
The number of blocks of $A$ that still remain in this region 
is over $(c_1/2)\log (N/B)$ with probability at most $(N/B)^{-c_1}$,
provided $c_0\ge 3$.
\end{lemma}

\begin{proof}
The expected number of blocks that still remain is at most 
$(c_1/2^{2c_0})\log (N/B)$. 
Thus, if we let $X$ denote the number blocks that still remain, then,
by a Chernoff bound
(e.g., Lemma~\ref{lem:chernoff} from the appendix)
and the fact that $c_0\ge 3$,
we have he following:
\begin{eqnarray*}
\Pr\left( X > (c_1/2) \log (N/B)\right) 
&<& 
2^{-(c_1/2) \log (N/B) \, (2 c_0 - 3)} \\
&<& 
2^{-c_1\log (N/B)} \\
&=& (N/B)^{-c_1}.
\end{eqnarray*}
\end{proof}

So this lemma gives us a lower bound for $c_0$.
In addition,
if we take $c_1=d+2$, then the probability that there is
any of our regions of size $c_1\log (N/B)$ blocks with more than
$c_1\log (N/B)/2$ blocks still remaining is at most $(N/B)^{-(d+1)}$,
which gives us a high probability of success.

Since $c_1$ is a constant
and the number of blocks we can fit in memory is 
$m=M/B\ge \log^{\epsilon^2} (N/B)=\log^{\epsilon^2} n$, 
by our wide-block and tall-cache assumptions,
we can apply the data-oblivious
sorting algorithm of Lemma~\ref{lem:sub-sort}
to sort
each of the $O(n/log n)$ regions using $O(\log n)$ I/Os a piece, 
putting every
distinguished element before any unmarked elements.
Thus, the overall number of I/Os for such a step is linear.
Then, we compact each such region to its first $(c_1/2)\log n$ blocks.
This action therefore halves the size of the array, $A$. 
So we then repeat the above actions for the smaller array.
We continue these reductions until the number of the remaining set of
blocks is less than $n/\log^2_{m} n$, at which point we
completely compress the remaining array $A$ by using the
data-oblivious sorting algorithm of
Lemma~\ref{lem:sub-sort}
on the entire array.
This results in array of size at most $r$, which we concatenate to
the array, $C$, of size $4r$, to produce a compacted array of size
$5r$.
Therefore, we have the following.

\begin{theorem}
\label{thm:network2}
Given an array, $A$, of size $N$, such that at most $R<N/4$ 
of $A$'s elements are
marked as distinguished, we can compress the distinguished elements in $A$
into an array $B$ of size $5R$ 
using a data-oblivious algorithm that uses
$O(N/B)$ I/Os and succeeds with probability at
least $1-1/(N/B)^d$, for any given constant $d\ge 1$,
assuming $B\ge \log^\epsilon (N/B)$ and $M\ge B^{1+\epsilon}$,
for some constant $\epsilon>0$.
\end{theorem}

\begin{proof}
To establish the claim
on the number of I/Os
for this algorithm, note
that the I/Os needed for all the iterations
are proportional to the geometric sum,
\[
n + n/2 + n/4 + \cdots + n/\log^2_m n,
\]
which is $O(n)=O(N/B)$.
In addition, when we complete the compression using the deterministic
sorting algorithm of Lemma~\ref{lem:sub-sort},
we use 
\[
O((n/\log^2_m n)\log^2_m (n/\log^2_m n))= O(n)=O(N/B)
\]
I/Os, under our weak wide-block and tall-cache assumptions.
Thus, the entire algorithm uses $O(N/B)$ I/Os and succeeds with very
high probability.
\end{proof}

Incidentally, we can also show that it is also possible to perform 
loose compaction
without the wide-block and tall-cache assumptions, albeit with
a slightly super-linear number of I/Os.
Specifically, we prove the following in the full version of this
paper.

\begin{theorem}
\label{thm:logstar}
Given an array, $A$, of size $N$, such that at most $R<N/4$ 
of $A$'s elements are
marked as distinguished, we can compress the distinguished elements in $A$
into an array $B$ of size $4.25R$ 
using a data-oblivious algorithm that uses
$O((N/B)\log^* (N/B))$ I/Os and succeeds with probability at
least $1-1/(N/B)^d$, for any given constant $d\ge 1$,
assuming only that $B\ge 1$ and $M\ge 2B$.
\end{theorem}

\section{Selection and Quantiles}

\paragraph{Selection.}
Suppose we are given an array $A$ of $n$ comparable items and want to
find the $k$th smallest element in $A$.
For each element $a_i$ in $A$, 
we mark $a_i$ as distinguished with probability
$1/n^{1/2}$. Assuming that there are at most $n^{1/2}+n^{3/8}$ distinguished
items, we then compress all the distinguished items in $A$
to an array, $C$, of size $n^{1/2}+n^{3/8}$ 
using the method of Theorem~\ref{thm:tight-sparse},
which runs in $O(n)$ time in this case and, as we prove,
succeeds with high probability.
We then sort the items in $C$ using a data-oblivious algorithm,
which can be done in $O(n^{1/2}\log^2 n)$ time by Lemma~\ref{lem:sub-sort},
considering
empty cells as holding $+\infty$.
We then scan this sorted array, $C$.
During this scan, we save in our internal registers,
items, $x'$ and $y'$, with ranks
$\lceil k/n^{1/2}\,-\,n^{3/8}\rceil$ 
and $|C| - \lceil (n-k)/n^{1/2}\,-\,2n^{3/8}\rceil$, respectively, in $C$, 
if they exist.
If $x'$ (resp., $y'$) does not exist, then we set $x'=-\infty$
(resp., $y'=+\infty$).
We then scan $A$ to find $x''$ and $y''$,
the smallest and largest elements in $A$, respectively.
Then we set $x=\max\{x',x''\}$ and $y=\min\{y',y''\}$.
We show below that, w.v.h.p.,
the $k$th smallest element
is contained in the range $[x,y]$ and there are $O(n^{7/8})$ items 
of $A$ in this range.

\begin{lemma}
\label{lem:sample}
There are more than $n^{1/2}+n^{3/8}$ elements in $C$ with probability at
most $e^{-n^{1/4}/3}$, and fewer than $n^{1/2}-n^{3/8}$ elements in
$C$ with probability at most $e^{-n^{1/4}/2}$.
\end{lemma}

\begin{proof}
Let $X$ denote the number of elements of $A$ chosen to belong to $C$.
Noting that each element of $A$ is chosen independently with
probability $1/n^{1/2}$ to belong to $C$,
by a standard Chernoff bound (e.g., see~\cite{mu-pcrap-05}, Theorem 4.4),
\begin{eqnarray*}
\Pr(X>n^{1/2}+n^{3/8}) &=& \Pr(X>(1+n^{-1/8})n^{1/2}) \\
	&\le& e^{-n^{1/4}/3} ,
\end{eqnarray*}
since $E(X)=n^{1/2}$.
Also, by another standard Chernoff bound
(e.g., see~\cite{mu-pcrap-05}, Theorem 4.5),
\begin{eqnarray*}
\Pr(X<n^{1/2}-n^{3/8}) &=& \Pr(X<(1-n^{-1/8})n^{1/2}) \\
	&\le& e^{-n^{1/4}/2} .
\end{eqnarray*}
\end{proof}

Recall that we then performed a scan, where we save in our internal registers,
items, $x'$ and $y'$, with ranks
$\lceil k/n^{1/2}\,-\,n^{3/8}\rceil$ 
and $|C| - \lceil (n-k)/n^{1/2}\,-\,n^{3/8}\rceil$, respectively, in $C$, 
if they exist.
If $x'$ (resp., $y'$) does not exist, then we set $x'=-\infty$
(resp., $y'=+\infty$).
We then scan $A$ to find $x''$ and $y''$,
the smallest and largest elements in $A$, respectively.
Then we set $x=\max\{x',x''\}$ and $y=\min\{y',y''\}$.

\begin{lemma}
\label{lem:select}
With probability at least 
$1-2e^{-n^{1/8}/9}-e^{-4n^{3/8}/5}-e^{-n^{1/4}/3}-e^{n^{1/4}/2}$, 
the $k$th smallest element of $A$ is contained in the range $[x,y]$ and
there are at most $8n^{7/8}$ items of $A$ in the range $[x,y]$.
\end{lemma}

\begin{proof}
Let us first suppose that $k\le 2n^{7/8}$.
In this case, we pick the smallest element in $A$ for $x$, in which case
$k\ge x$.

So, to consider the remaining possibility that $x>k$, for $k > 2n^{7/8}$.
We can model the number of elements of $A$ less than or equal to $x$ as 
the sum, $X$, of $k'=k/n^{1/2}-n^{3/8}$ 
independent geometric random variables with
parameter $p=n^{-1/2}$.
The probability that $k$ is less than or equal to $x$ is bounded by
\[
\Pr(X > k) = \Pr(X > (n^{1/2} +  t)k'),
\]
where $t=n^{7/8}/(k/n^{1/2}-n^{3/8})$; hence,
$t \ge n^{3/8}$ and $k'>n^{3/8}$.
If $t/n^{1/2}<1/2$, then,
by a Chernoff bound 
(e.g., Lemma~\ref{lem:chernoff2} from the appendix), we have the following: 
\begin{eqnarray*}
\Pr(X > k) &<& e^{-(t^2/n)k'/3} \\
           &<& e^{-n^{1/8}/3}.
\end{eqnarray*}
Similarly, if $t/n^{1/2}\ge 1/2$, then
\begin{eqnarray*}
\Pr(X > k) &<& e^{-(t/n^{1/2})k'/9} \\
           &<& e^{-n^{1/8}/9}.
\end{eqnarray*}
Thus, in either case, the probability that the $k$th smallest element
is greater than $x$ is bounded by this latter probability.
By a symmetric argument, assuming $|C|\ge n^{1/2}-n^{3/8}$,
the probability that there are the more than 
$(n-k)$ elements of $A$ greater than $y$,
for $k<n-2n^{7/8}$, 
is also bounded by this latter probability 
(note that, for $k\ge n-2n^{7/8}$, it is trivially true that 
the $k$th smallest element is less than or equal to $y$).
So, with probability at least $1-2e^{-n^{1/8}/9}-e^{-n^{1/4}/2}$, 
the $k$th smallest
element is contained in the range $[x,y]$.

Let us next consider the number of elements of $A$ in the range $[x,y]$.
First, note that, probability at least $1-e^{n^{1/4}/3}$,
there are at most $n'=4n^{3/8}$ elements of the random sample, $C$, in this
range; hence, we can model the number of elements of
$A$ in this range as being bounded by the sum, $Y$, of $n'$ geometric
random variables with parameter $p=1/n^{1/2}$.
Thus, by a Chernoff bound
(e.g., Lemma~\ref{lem:chernoff2} from the appendix), we have the following: 
\begin{eqnarray*}
\Pr(Y > 8n^{7/8}) &=& \Pr(Y > (n^{1/2} + n^{1/2})n') \\
	   &<& e^{-n'/5} \\
           &=& e^{-4n^{3/8}/5}.
\end{eqnarray*}
This gives us the lemma.
\end{proof}

So, we make an addition scan of $A$ to mark the elements in $A$ that
are in the range $[x,y]$, and we then compress these items to an
array, $D$, of size $O(n^{7/8})$, using the method of
Theorem~\ref{thm:tight-sparse}, which runs in $O(n)$ time in this case.
In addition, we can determine the rank, $r(x)$, of $x$ in $A$. Thus,
we have just reduced the problem to returning the item in $D$
with rank $k-r(x)+1$. We can solve this problem by sorting $D$ using
the oblivious sorting method of Lemma~\ref{lem:sub-sort}
followed by a scan to obliviously select the item with this rank.
Therefore, we have the following:

\begin{theorem}
\label{thm:select}
Given an integer, $1\le k\le n$,
and an array, $A$, of $n$ comparable items, we can select the $k$th
smallest element in $A$ in $O(n)$ time using a data-oblivious algorithm
that succeeds with probability at least $1-n^{-d}$, for any given
constant $d>0$.
\end{theorem}

Note that the running time (with a very-high success probability)
of this method beats the $\Omega(n\log\log n)$
lower bound of Leighton {\it et al.}~\cite{lms-pnsms-97},
which applies to any high-success-probability randomized 
data-oblivious algorithm based on the exclusive use of
compare-exchange as the primitive data-manipulation operation.
Our method is data-oblivious, but it also uses primitive operations
of copying, summation, and random hashing.
Thus, Theorem~\ref{thm:select} demonstrates the power of using these
primitives in data-oblivious algorithms.
In addition, by substituting data-oblivious 
external-memory compaction and sorting steps for the internal-memory
methods used above, we get the following:

\begin{theorem}
\label{thm:select2}
Given an integer, $1\le k\le N$,
and an array, $A$, of $N$ comparable items, we can select the $k$th
smallest element in $A$ using $O(N/B)$ I/Os with a data-oblivious 
external-memory algorithm
that succeeds with probability at least $1-(N/B)^{-d}$, for any given
constant $d>0$, assuming only that $B\ge 1$ and $M\ge 2B$.
\end{theorem}

\paragraph{Quantiles.}
Let us now consider the problem of selecting $q$ quantile
elements from an array $A$ using an external-memory data-oblivious
algorithm, for the case when $q\le (M/B)^{1/4}$, which will be
sufficient for this algorithm to prove useful for our external-memory
sorting algorithm, which we describe in Section~\ref{sec:sort}.

If $(M/B)> (N/B)^{1/4}$, then we sort $A$ using the deterministic
data-oblivious algorithm of Lemma~\ref{lem:sub-sort},
which uses $O(N/B)$ I/Os in this case. Then we simply read out the
elements at ranks that are multiples of $N/(q+1)$, rounded to integer
ranks.

Let us therefore suppose instead that $(M/B)\le (N/B)^{1/4}$;
so $q\le (N/B)^{1/16}$.
In this case, we randomly choose each element of $A$ to belong to a
random subset, $C$, with probability $1/N^{1/4}$.
With high probability, there are at most $N^{3/4}\pm N^{1/2}$ 
such elements, so we can compact them into an array, $C$, of capacity
$N^{3/4}+N^{1/2}$ by Theorem~\ref{thm:tight-sparse}, using $O(N/B)$ I/Os.
We then sort this array and compact it down to capacity $N^{3/4}+N^{1/2}$,
using Lemma~\ref{lem:sub-sort}, and let $|C|$ denote its actual size (which
we remember in private memory).
We then scan
this sorted array, $C$, and read into Alice's memory every 
item, $x_i$, with rank 
that is a multiple of
$({\hat n}/(q+1)) - N^{1/2}$, rounded to integer ranks, where 
${\hat n} = N^{3/4}$, with the exception that $x_1$ is the smallest element
in $A$. 
We also select items at ranks,
$y_i=|C|-(N^{3/4}-N^{3/4}i/(q+1) - 2N^{1/2})$, rounded to integer ranks,
with the exception that $y_q$ is the largest element in $A$.
Let $[x_i,y_i]$ denote each such pair of such items,
which, as we show below, will surround a value,
${\hat n}/(q+1)$, with high probability.
In addition, as we also show in the analysis,
there are at most $8N^{3/4}$ elements of $A$ in
each interval $[x_i,y_i]$, with high probability.
That is, there are at most $O(N^{13/16})$ elements of $A$ in any
interval $[x_i,y_i]$, with high probability.
Storing all the $[x_i,y_i]$ intervals in Alice's memory,
we scan $A$ to identify for each element in $A$ if it is contained
in such an interval $[x_i,y_i]$, marking it with $i$ in this case,
or if it is outside every such interval.
During this scan we also maintain counts (in Alice's memory)
of how many elements of $A$
fall between each consecutive pair of intervals, $[x_i,y_i]$ and
$[x_{i+1}, y_{i+1}]$, and how many elements fall inside each 
interval $[x_i,y_i]$.
We then compact all the elements of $A$ that are inside such
intervals into an array $D$ of size
$O(N^{13/16})$ using Theorem~\ref{thm:tight-sparse}, 
using $O(N/B)$ I/Os,
and we pad this array with dummy
elements so the number of elements of $D$ in each interval
$[x_i,y_i]$ is exactly $\lceil 8N^{3/4}\rceil$.
We then
sort $D$ using the data-oblivious method of Lemma~\ref{lem:sub-sort}.
Next, for each subarray of $D$ of size $\lceil 8N^{3/4}\rceil$ we use
the selection algorithm of Theorem~\ref{thm:select2}
to select the $k_i$th smallest
element in this subarray, where $k_i$ is the value that will return
the $i$th quantile for $A$ (based on the counts we computed during
our scans of $A$).

\begin{lemma}
The number of elements of $A$ in $C$
is more than $N^{3/4}+N^{1/2}$ with probability
at most $e^{-N^{1/4}/3}$,
and the number of elements of $A$ in $C$ is less than 
$N^{3/4}-N^{1/2}$ with probability
at most $e^{-N^{1/4}/2}$.
\end{lemma}
\begin{proof}
Let $X$ denote the number of elements of $A$ chosen to belong to $C$.
Noting that each element of $A$ is chosen independently with
probability $1/N^{1/4}$ to belong to $C$,
by a standard Chernoff bound (e.g., see~\cite{mu-pcrap-05}, Theorem 4.4),
\begin{eqnarray*}
\Pr(X>N^{3/4}+N^{1/2}) &=& \Pr(X>(1+N^{-1/4})N^{3/4}) \\
	&\le& e^{-N^{1/4}/3} ,
\end{eqnarray*}
since $E(X)=n^{1/2}$.
Also, by another standard Chernoff bound
(e.g., see~\cite{mu-pcrap-05}, Theorem 4.5),
\begin{eqnarray*}
\Pr(X<N^{3/4}-N^{1/2}) &=& \Pr(X<(1-N^{-1/4})N^{3/4}) \\
	&\le& e^{-N^{1/4}/2} .
\end{eqnarray*}
\end{proof}

We then compress these elements
into an array, $C$, of size
assumed to be at most
$N^{3/4}+N^{1/2}$ by Theorem~\ref{thm:tight-sparse}, using $O(N/B)$ I/Os,
which will fail with probability 
at most $1/N^{3c/4}$, for any given constant $c>0$.
Given the array $C$, we select items at ranks 
$x_i=({\hat n}i/(q+1)) - N^{1/2}$, rounded to integer ranks.
where ${\hat n}=N^{3/4}$.
We also select items at ranks,
$y_i=|C|-(N^{3/4}-N^{3/4}i/(q+1) - 2N^{1/2})$, rounded to integer ranks.
Let $[x_i,y_i]$ denote each such pair, 
with the added convention that
we take $x_1$ to be the smallest element in $A$ and
$y_q$ to be the largest element in $A$.

\begin{lemma}
There are more than
$8N^{3/4}$ elements of $A$ in
any interval $[x_i,y_i]$
with probability at most $e^{-N^{1/2}/9}+2e^{-N^{1/4}/3}$.
\end{lemma}
\begin{proof}
Let us assume that $N^{3/4}-N^{1/2}\le |C|\le N^{3/4}+N^{1/2}$,
which hold with probability at least $1-2e^{-N^{1/4}/3}$.
Thus, there are at most $4N^{1/2}$ elements in $C$ that
are in any $[x_i,y_i]$ pair, other than the first or last. Thus, in such
a general case, the number of elements, $X$, from $A$ in
this interval has expected value $E(X)\le 4N^{3/4}$.
In addition, since $X$ is the sum of geometric random variables with
parameter $1/N^{1/4}$,
\begin{eqnarray*}
\Pr(X>8N^{3/4}) &=& \Pr(X > (N^{1/4} + N^{1/4})4N^{1/2}) \\
		&\le& e^{-N^{1/2}/9}.
\end{eqnarray*}
The probability bounds for the first and last intervals are proved
by a similar argument.
\end{proof}

In addition, note that there are at most $(N/B)^{1/16}$ intervals,
$[x_i,y_i]$.
Also, we have the following.

\begin{lemma}
Interval $[x_k,y_k]$ contains the $k$th quantile with
probability at least $1-2e^{-N^{1/4}/3}$.
\end{lemma}
\begin{proof}
Let us consider the probability that the $k$th quantile is less than $x_k$.
In the random sample, $x_k$ has rank 
${\hat n}k/(q+1)-N^{1/2}=N^{3/4}k/(q+1)-N^{1/2}$.
Thus, the number, $X$, of elements from $A$ less than this number has 
expected value $E(X)=Nk/(q+1)-N^{3/4}$.
Since $X$ is the sum of geometric random variables with parameter $N^{-1/4}$,
 we can bound $\Pr\left(X > Nk/(q+1)\right)$ by
{\small
\begin{eqnarray*}
 &\relax& 
    \Pr\left( X> (N^{1/4} + \tau)\cdot
       \left(\frac{N^{3/4}k}{q+1}-N^{1/2}\right)\right) \\
     &\le& e^{-(\tau N^{-1/4})^2 (N^{3/4}k/(q+1)\, -\, N^{1/2})/3} \\
     &=& e^{-\tau^2 N^{-1/2} (N^{3/4}k/(q+1)\, -\, n^{1/2})/3},
\end{eqnarray*}
}
where $\tau = N^{3/4}/(N^{3/4}k/(q+1)\, -\, N^{1/2})$, which greater than $1$.
So
\begin{eqnarray*}
\Pr(X > Nk(q+1)) &\le& e^{-\tau N^{-1/2}N^{3/4}/3} \\
                 &<& e^{-N^{1/4}/3} .
\end{eqnarray*}
By a symmetric argument, the $k$th quantile is more than $y_k$ by this same
probability.
\end{proof}

Thus, each of the intervals, $[x_k,y_k]$, contains 
the $k$th quantile with
probability at least $1-2qe^{-N^{1/4}/3}$.
Therefore, we have the following.

\begin{theorem}
\label{thm:quantiles}
Given an array, $A$, of $N$ comparable items, we can select the 
$q\le (M/B)^{1/4}$ quantiles
in $A$ using $O(N/B)$ I/Os with a data-oblivious 
external-memory algorithm
that succeeds with probability at least $1-(N/B)^{-d}$, for any given
constant $d>0$, assuming only that $B\ge 1$ and $M\ge 2B$.
\end{theorem}

\section{Data-Oblivious Sorting}
\label{sec:sort}
Consider now the sorting problem, where we are given an array, $A$, of $N$
comparable items stored as key-value pairs and we want to output an array $C$
of size $N$ holding the items from $A$ in key order.
For inductive reasons, however, we allow both $A$ and $C$ to be arrays of
size $O(N)$ that have $N$ non-empty cells, with the requirement
that the items in non-empty
cells in the output array $C$ be given in non-decreasing order.

We begin by computing $q$ quantiles for the items in $A$ using
Theorem~\ref{thm:quantiles}, for $q=(M/B)^{1/4}$.
We then would like to distribute the items of $A$ distributed between all
these quantiles to $q+1$ subarrays of size $O(N/q)$ each.
Let us
think of each subgroup (defined by the quantiles)
in $A$ as a separate color, so that each item in $A$
is given a specific color, $1,2,\ldots,q+1$, with there being 
$\lceil N/(q+1)\rceil$ items
of each color.

\paragraph{Multi-way Data Consolidation.}
To prepare for this distribution, we do a $(q+1)$-way consolidation of $A$, 
so that the items in each block of size $B$ in 
the consolidated array $A'$ are all of the same color.
We perform this action as follows.
Read into Alice's memory the first $q+1$ blocks of $A$, 
and separate them into $q+1$ groups,
one for each color. 
While there is a group of items with the same color
of size at least $B$, output a block of these items to $A'$.
Once we have output $q'$ such blocks, all the remaining colors in Alice's memory
have fewer than $B$ members. So we then output $q+1-q'$ empty blocks to $A'$,
keeping the ``left over'' items in Alice's memory.
We then repeat this computation for the next $q+1$ blocks of $A$, and the
next, and so on.
When we complete the scan of $A$, all the blocks in $A'$ will be completely
full of items of the same color or they will be completely empty.
We finish this $(q+1)$-way consolidation, then, by outputting $(q+1)$ blocks,
each containing as many items of the same color as possible.
These last blocks are the only partially-full blocks in $A'$.
Thus, all the blocks of $A'$ are monochromatic and
all but these last blocks of $A'$ are full.

\paragraph{Shuffle-and-Deal Data Distribution.}
Our remaining goal, then, is to distribute the (monochromatic) blocks of $A'$
to $q+1$ separate arrays, $C_1$, $\ldots$, $C_{q+1}$, one for each color.
Unfortunately, doing this as a straightforward scan of $A'$
may encounter the colors in a non-uniform fashion.
To probabilistically avoid this ``hot spot'' behavior, we apply a
\emph{shuffle-and-deal} technique, where we perform a random permutation to
the $n'=O(N/B)$ blocks in $A'$ in a fashion somewhat reminiscent of
Valiant-Brebner routing~\cite{vb-uspc-81}.
The random permutation algorithm we use here is the well-known algorithm
(e.g., see Knuth~\cite{k-sa-98}),
where, for $i=1,2,\ldots,n'$, we swap block $i$ and a random block 
chosen uniformly from the range $[i,n']$.
This is the ``shuffle,''
and even though the adversary, Bob, can see us perform this shuffle, note that
the choices we make do not depend on data values. Given this
shuffled deck of blocks in $A'$, we then perform a series of scans to
``deal'' the blocks of $A'$ to the $q+1$ arrays.

We do this ``deal'' as follows.
We read in the next $(M/B)^{3/4}$ blocks of $A'$.
Note that, w.v.h.p., there should now be at most $O((M/B)^{1/2})$ blocks of each
color now in Alice's memory. So we write out $c(M/B)^{1/2}$ blocks to each
$C_i$ in Bob's memory, 
including as many full blocks as possible and padding with empty
blocks as needed (to keep accesses being data-oblivious), for a constant $c$
determined in the analysis. 
Then, we apply Theorem~\ref{thm:network2},
which implies a success with high probability,
to compact each $C_i$ to
have size $O(N/q)=O(N/(M/B)^{1/4})$ each, 
which is $O(N/(qB))$ blocks.
We then repeat the above computation for each subarray, $C_i$.

\paragraph{Data-Oblivious Failure Sweeping.}
We continue in this manner until we have formed $O(n^{1/2})$
subarrays of size $O(n^{1/2})$ each, where $n=N/B$. 
At this point, we then recursively call our sorting
algorithm to produce a padded sorting of each of the subarrays.
Of course, since we are recursively solving smaller problems, whose
success probability depends on their size, some of
these may fail to correctly produce a padded sorting of their inputs.
Let us assume that at most $O(n^{1/4})$ of these recursive sorts
fail, however, and apply Theorem~\ref{thm:network} 
to deterministically compact all of these subarrays into a single array,
$D$, of size $O(n^{3/4})$, in $O(n\log_m n)$ time, where $m=M/B$.
We then apply the deterministic data-oblivious sorting method of
Lemma~\ref{lem:sub-sort}
to $D$, and then we perform a
reversal of Theorem~\ref{thm:network}
to expand these sorted elements back to their
original subarrays. 
This provides a data-oblivious version of
the failure sweeping technique~\cite{gg-frpmp-97}
and gives us a padded sorting of $A$ w.v.h.p. 

Finally, after we have completed the algorithm for producing a padded
sorting of $A$, we perform a tight order-preserving compaction for
all of $A$ using Theorem~\ref{thm:network}. 
Given appropriate probabilistic guarantees, given below,
this completes the algorithm.

\begin{lemma}
Given $(M/B)^{3/4}$ blocks of $A'$, read in from consecutive blocks from a
random permutation, more than $c(M/B)^{1/2}$ of these blocks 
are of color, $\chi$, for a given color, $\chi$, with probability
less than $(N/B)^{-d}$, for $c>2de/\epsilon^2$,
where $\epsilon$ is the constant used in the
wide-block and tall-cache assumptions.
\end{lemma}
\begin{proof}
Let $X$ be the number of blocks among $(M/B)^{3/4}$ blocks
chosen independently without replacement from $A'$ that are of color
$\chi$ and let
$Y$ be the number of blocks among $(M/B)^{3/4}$ blocks
chosen independently with replacement from $A'$ that are of color
$\chi$.
By a theorem (4) of Hoeffding~\cite{h-pisbrv-63},
\[
\Pr(X > c(M/B)^{1/2}) \le \Pr(Y > c(M/B)^{1/2}).
\]
Thus, since $E(Y)=(M/B)^{1/2}$, 
then, by a Chernoff bound (e.g., Lemma~\ref{lem:chernoff} from the appendix),
\[
\Pr(Y > c(M/B)^{1/2}) \le 2^{-c(M/B)^{1/2}} \le (N/B)^d ,
\]
for $c>2de/\epsilon^2$.
\end{proof}

This bound applies to any $(M/B)^{3/4}$ blocks 
we read in from $A'$, and any specific color, $\chi$.
Thus, we have the following.

\begin{corollary}
For each set of $(M/B)^{3/4}$ blocks of $A'$, read in 
from consecutive blocks from the constructed
random permutation, more than $c(M/B)^{1/2}$ of these blocks 
are of any color, $\chi$, with probability
less than $(N/B)^{-d}$, for $c>2d'e/\epsilon^2$,
where $\epsilon$ is the constant used in the
wide-block and tall-cache assumptions and $d'\ge d+1$.
\end{corollary}
\begin{proof} 
There are 
$(N/B)/(M/B)^{3/4}$ sets of 
$(M/B)^{3/4}$ blocks of $A'$ read in 
from consecutive blocks from the constructed
random permutation.
In addition, there are 
$(M/B)^{1/16}$ colors.
Thus, by the above lemma and the union bound,
the probability that any of these sets overflow a color $\chi$ is at
most
\[
\frac{N/B}{(M/B)^{3/4}}
\cdot
\frac{1}{(N/B)^{d'}}
\cdot
(M/B)^{1/16} 
\le 
(N/B)^{-d},
\]
for $d'\ge d+1$.
\end{proof}

So we write out $c(M/B)^{1/2}$ blocks to each
$C_i$, including as many full blocks as possible and padding with empty
blocks as needed. 
We then repeat the quantile computation and shuffle-and-deal 
computation on each of the $C_i$'s.

At the point when subproblem sizes become of size $O(n^{1/2})$, we
then switch to a recursive computation, which we claim inductively
succeeds with probability $1-1/n^{d/2}$, for any given constant $d\ge 2$,
for $n=N/B$.

\begin{lemma}
There are more than $n^{1/4}$ failing recursive subproblems with
probability at most $2^{-n^{1/4}}$.
\end{lemma}
\begin{proof}
Each of the recursive calls fails independently.
Thus, if we let $X$ denote the number of failing recursive calls,
then $E(X)\le n^{1/2}/n^{d/2}=n^{-(d-1)/2}$.
Thus, by a Chernoff bound (e.g., Lemma~\ref{lem:chernoff} from the appendix),
\begin{eqnarray*}
\Pr(X > n^{1/4}) &=& \Pr(X > n^{1/4\,+\,(d-1)/2}\cdot n^{(d-1)/2}) \\
	&\le& 2^{-n^{1/4}} .
\end{eqnarray*}
\end{proof}

Thus, the failure sweeping step in our sorting algorithm succeeds
with high probability, which completes the analysis and gives us
the following.

\begin{theorem}
Given an array, $A$, of size $N$, we can perform a data-oblivious 
sorting of $A$ with an algorithm that 
succeeds with probability $1-1/(N/B)^d$
and uses $O((N/B)\log_{M/B} (N/B))$ I/Os,
for any given
constant $d\ge 1$, assuming that $B\ge \log^\epsilon (N/B)$ and 
$M\ge B^{1+\epsilon}$, for some small constant $\epsilon>0$.
\end{theorem}

\ifFull
\section{Conclusion}
We have given several data-oblivious algorithms for fundamental
combinatorial problems involving outsourced data. 
It would be interesting to know if a linear-I/O
data-oblivious algorithm is possible for data compression without the
wide-block and tall-cache assumptions
(in Appendix~\ref{app:logstar}, we show how to relax
this wide-block assumption, albeit at the cost of an extra log-star factor
in the I/O performance).
It would also be interesting to investigate
data-oblivious algorithms for graph and geometric problems in this
same model, as such computations are well-motivated for large
outsourced data sets.
\fi

\subsection*{Acknowledgments}
We would like to thank Pawel Pszona for some helpful comments
regarding an earlier version of this paper.
This research was supported in part by the National Science
Foundation under grants 0724806, 0713046, 0847968, and 0953071.

{\raggedright
\bibliographystyle{abbrv}
\bibliography{cuckoo,extra,goodrich,geom,cuckoo2}
}

\clearpage
\begin{appendix}
\section{Some Chernoff Bounds}
Several of our proofs make use of Chernoff bounds
\ifFull
(e.g., see~\cite{as-pm-92,c-maeth-52,mu-pcrap-05,mr-ra-95,m-cgitr-93}
\else
(e.g., see~\cite{mu-pcrap-05}
\fi
for other examples), which, for the sake of completeness,
we review in this subsection.
We begin with the following
Chernoff bound, which is 
a simplification of a well-known bound.

\begin{lemma}
\label{lem:chernoff}
Let $X=X_1+X_2+\cdots+X_n$ be the sum of 
independent $0$-$1$ random variables, such that $X_i=1$ with probability
$p_i$, and let $\mu\ge E(X)=\sum_{i=1}^n p_i$.
Then, for $\gamma > 2e$,
\[
\Pr\left( X > \gamma\mu\right) < 2^{-\gamma\mu\log (\gamma/e)} .
\]
\end{lemma}
\begin{proof}
By a standard Chernoff bound 
\ifFull
(e.g., see~\cite{as-pm-92,c-maeth-52,mu-pcrap-05,mr-ra-95,m-cgitr-93})
\else
(e.g., see~\cite{mu-pcrap-05})
\fi
and the fact that $\gamma > 2e$,
\begin{eqnarray*}
\Pr\left( X > \gamma\mu\right) &<& 
  \left( \frac{e^{\gamma-1}}{\gamma^\gamma} \right)^\mu \\
  &\le& \left( \frac{e}{\gamma} \right)^{\gamma\mu} \\
  &=& 2^{-\gamma\mu\log (\gamma/e)} .
\end{eqnarray*}
\end{proof}

There are other 
simplified Chernoff
bounds similar to that of Lemma~\ref{lem:chernoff}
\ifFull
(e.g., see~\cite{as-pm-92,mu-pcrap-05,mr-ra-95,m-cgitr-93}),
\else
(e.g., see~\cite{mu-pcrap-05,m-cgitr-93}),
\fi
but they typically omit the $\log (\gamma/e)$ term
(where the log is base-2, of course). 
We include it here, since it is
useful for large $\gamma$, which will be the case for some of our uses.
Nevertheless, we sometimes
leave off the $\log (\gamma/e)$ term, as well, in applying
Lemma~\ref{lem:chernoff}, if that
aids simplicity, since $\log (\gamma/e) > 1$ for $\gamma > 2e$.

In addition, we also need a Chernoff bound for the sum, $X$, of $n$
independent geometric random variables with parameter $p$, that is,
for $X$ being a negative binomial random variable with parameters
$n$ and $p$.
Recall that a geometric random variable with parameter $p$ is 
a discrete random variable that is equal
to $j$ with probability $q^{j-1}p$, where $q=1-p$.
Thus, $E(X) = \alpha n$, where $\alpha=1/p$.

\begin{lemma}
\label{lem:chernoff2}
Let $X=X_1+X_2+\cdots+X_n$ 
be the sum of $n$ independent geometric random variables with
parameter $p$.
Then we have the following:
\begin{itemize}
\item If $0<t < \alpha/2$, then $ \Pr(X > (\alpha + t)n) \le  e^{-(tp)^2n/3}$. 
\item If $t\ge \alpha/2$, then $ \Pr(X > (\alpha + t)n) \le  e^{-tpn/9} .  $
\item If $t\ge \alpha$, then $ \Pr(X > (\alpha + t)n) \le  e^{-tpn/5} .  $
\item If $t\ge 2\alpha$, then $ \Pr(X > (\alpha + t)n) \le  e^{-tpn/3} .  $
\item If $t\ge 3\alpha$, then $ \Pr(X > (\alpha + t)n) \le  e^{-tpn/2} .  $
\end{itemize}
\end{lemma}
\begin{proof}
We follow the approach of Mulmuley~\cite{m-cgitr-93}, who 
uses the Chernoff technique 
\ifFull
(e.g., see~\cite{c-maeth-52,mu-pcrap-05,mr-ra-95,m-cgitr-93})
\else
(e.g., see~\cite{mu-pcrap-05})
\fi
to prove a similar result for the special case when $p=1/2$ 
and $t\ge 6$ (albeit with a slight flaw, which we fix).
For $0< \lambda < \ln (1/(1-p))$, 
\begin{eqnarray*}
E\left(e^{\lambda X_i}\right) &=&
  \sum_{j=1}^{\infty} e^{\lambda j} \Pr(X_i=j) \\
  &=& \sum_{j=1}^{\infty} e^{\lambda j} q^{j-1} p \\
  &=& p e^\lambda \sum_{j=0}^{\infty} (e^\lambda q)^{j} \\
  &=& \frac{p e^\lambda}{1-e^\lambda q}.
\end{eqnarray*}
Applying the Chernoff technique, then,
\begin{eqnarray*}
\Pr(X > (\alpha + t)n) &\le& e^{-\lambda (\alpha+t)n} 
			 \left(\frac{p e^\lambda}{1-e^\lambda q} \right)^n \\
   &=& p^n \left(\frac{e^{-\lambda(\alpha + t -1)}}{1-e^\lambda q} \right)^n.\\
\end{eqnarray*}
Let $\beta=p/(1-p)$ and observe that we can satisfy the condition that
$0< \lambda < \ln (1/(1-p))$ by setting 
\[
e^\lambda = 1 + \frac{\beta t}{\alpha + t} .
\]
By substitution and some calculation, note that
\[
e^{-\lambda} = 1 - \frac{\beta t}{\alpha + t + \beta t} 
\]
and
\[
1-e^\lambda q = \frac{p}{1+tp} .
\]
Thus,
we can bound $\Pr(X > (\alpha + t)n)$ by  
\begin{eqnarray*}
&\relax& p^n 
\left( 1 - \frac{\beta t}{\alpha + t + \beta t} \right)^{(\alpha+t-1)n}
\left( \frac{1+tp}{p} \right)^n \\
&=& \left( 1 - \frac{\beta t}{\alpha + t + \beta t} \right)^{(\alpha+t-1)n}
\left( 1+tp \right)^n.
\end{eqnarray*}
Moreover, since $1-x\le e^{-x}$, for all $x$, 
\[
\left( 1 - \frac{\beta t}{\alpha + t + \beta t} \right)^{(\alpha+t-1)}
\le e^{-\frac{\beta t (\alpha + t -1)}{\alpha + t + \beta t}}
= e^{-tp} .
\]
Therefore, 
\[
\Pr(X > (\alpha + t)n) \le
e^{-tpn}
\left( 1+tp \right)^n.
\]
Unfortunately, if we use the well-known inequality, 
$1+x\le e^x$, with $x=tp$, 
to bound the ``$1+tp$'' term in the above equation,
we get a useless result.
So, instead, we use better approximations:
\begin{itemize}
\item
If $0< x < 1$, then we can use a truncated Maclaurin series
to bound
\[
\ln (1 + x) \le x - \frac{x^2}{2} + \frac{x^3}{3} .
\]
Thus,
\[
1 + x \le e^{x - \frac{x^2}{2} + \frac{x^3}{3}},
\]
which implies that, for $0<t<\alpha/2$,
\begin{eqnarray*}
\Pr(X > (\alpha + t)n) &\le&
      \left(\frac{e^{(tp)^3/3}}{e^{(tp)^2/2}}\right)^n \\
      &\le& e^{-(tp)^2n/3} .
\end{eqnarray*}
\item
The remaining bounds follow from the following facts, which are
easily verified:
\begin{enumerate}
\item If $x\ge 1/2$, then $1+x< e^{x/(1+1/8)}$.
\item If $x\ge 1$, then $1+x< e^{x/(1+1/4)}$.
\item If $x\ge 2$, then $1+x< e^{x/(1+1/2)}$.
\item If $x\ge 3$, then $1+x< e^{x/2}$.
\end{enumerate}
\end{itemize}
\end{proof}

Incidentally, the bound for $t\ge 3\alpha$ fixes a 
slight flaw in a Chernoff bound proof by Mulmuley~\cite{m-cgitr-93}.

\ifFull
\section{A Slightly Super-Linear Algorithm for Data-Oblivious Loose Compaction}
\label{app:logstar}
In this appendix, we provide a proof for
\textbf{Theorem~\ref{thm:logstar}}, which states that,
given an array, $A$, of size $N$, such that at most $R<N/4$ 
of $A$'s elements are marked as distinguished, we can compress 
the distinguished elements in $A$
into an array $B$ of size $4.25R$
using a data-oblivious algorithm that uses
$O((N/B)\log^* (N/B))$ I/Os and succeeds with probability at
least $1-1/(N/B)^d$, for any given constant $d\ge 1$,
assuming only that $B\ge 1$ and $M\ge 2B$.
This method does not necessarily 
preserve the order of the items in $A$.

\medskip
\begin{proof}
We begin by performing the consolidation operation of 
Lemma~\ref{lem:consolidate},
so that each block in $A$, save one, is completely full or 
completely empty. We
then perform the following RAM algorithm at the granularity of
blocks, so that each memory read or write is actually an I/O for an
entire block.
That is, we describe a RAM algorithm that runs in
$O(n\log^* n)$ time, with high probability, where $n=N/B$ and $r=R/B$,
and we implement
it for the external-memory model so as to perform each read or write
of a word as a read or write of a block of size $B$.

So, suppose we are given an array $A$ of size $n$, such that at most
$r<n/4$ of the elements in $A$ are marked as ``distinguished''
and
we want to map the distinguished elements from $A$ into an array $D$,
of size $4.25r$ using an algorithm that is data-oblivious.

Our algorithm begins by testing if 
$n$ is smaller than a constant, $n_0$, determined in the analysis,
in which case we compact
$A$ using the deterministic data-oblivious sorting algorithm
of Lemma~\ref{lem:sub-sort}.
In addition, we also test if
$r<n/\log^2 n$, in which case we
use Theorem~\ref{thm:tight-sparse} to compact $A$ into $D$.
(Note that in either of these base cases, we can compact the 
at most $r$ distinguished
elements of $A$ into an array size exactly $r$.)
Thus, let us assume that $r\ge n/\log^2 n$ and $n\ge n_0$.

Our algorithm for this general case
is loosely based on the parallel linear
approximate compaction algorithm of Matias and Vishkin~\cite{mv-chpnc-91}
and proceeds with a number of phases.
In this general case, we restrict $D$ to its first $4r$ cells, leaving
its last $0.25r$ cells for later use.
At the beginning of each phase, $i$, we assume inductively that there are
at most 
\[
\frac{r}{t_i^4}
\]
distinguished elements left in $A$, where $t_1 = 2^2$ and
\[
t_{i+1} = 2^{t_i},
\]
for $i\ge 1$.
That is, the $t_i$'s form the tower-of-twos sequence; hence, there are
at most $O(\log^* n)$ phases.

We begin our data-oblivious algorithm for compacting
the distinguished elements in $A$ into $D$ by performing $c_0$
$A$-to-$D$ thinning passes, for a constant, $c_0$, determined by the
following lemma.

\begin{lemma}
\label{lem:init}
Performing an initial $c_0$ $A$-to-$D$ thinning passes 
in the general-case algorithm
results in there being at most $r/t_1^4$
distinguished elements in $A$, with probability at 
least $1-n^{-d}$, for any given constant $d\ge1$, provided $c_0\ge 2^3$.
\end{lemma}
\begin{proof}
If we perform $c_0\ge 2^3$ initial $A$-to-$D$ thinning passes, then
the expected number of distinguished elements, 
$\mu$, still in $A$ will be bounded by
\[
\mu = \frac{r}{2^{16}},
\]
since the probability of a 
collision in each pass is at most $1/4$.
Moreover,
we can bound the number of failing distinguished elements by a sum, $X$, of
$r$ independent 0-1 random variables, each of which is $1$ with
probability $1/2^{16}$. 
Since $t_1=2^2=4$, 
\begin{eqnarray*}
\Pr\left( X > \frac{r}{t_1^{4}}\right) &=&
\Pr\left( X > \frac{r}{2^{8}}\right) \\
&=& \Pr\left( X > 2^8\cdot \frac{r}{2^{16}}\right) \\
&=& \Pr\left( X > 2^8\mu \right),
\end{eqnarray*}
which we can bound using the Chernoff bound of Lemma~\ref{lem:chernoff}
as
\[
\Pr\left( X > 2^8\mu \right) < 2^{-256\mu}.
\]
In addition, note that since $r\ge n/\log^2 n$,
\begin{eqnarray*}
256\mu &=& \frac{256r}{2^{16}} \\
     &\ge& \frac{256n}{2^{16}\log^2 n} \\
     &\ge& d\log n ,
\end{eqnarray*}
provided we set $n_0$ so that
the last inequality is true whenever $n\ge n_0$.
Thus, the probability of failure is at most $n^{-d}$.
\end{proof}
(for Lemma~\ref{lem:init})

So, let us assume inductively
that at the beginning of phase $i$ there are at most $r/t_i^4$
distinguished elements remaining in $A$.
If, at the beginning of phase $i$, 
\[
\frac{r}{t_i^4} \le \frac{n}{\log^2 n} ,
\]
then we apply Theorem~\ref{thm:tight-sparse} 
to compress the remaining distinguished
elements of $A$ into the last $0.25r\le r/t_i^4$ cells of $D$, which completes
the algorithm.
That is,
this results in a compaction of $A$ into an array $D$ of size $4.25r$.

If, on the other hand, $r/t_i^4 > n/\log^2 n$, then we continue with phase
$i$, which we divide into two steps, a thinning-out step and a
region-compaction step.

\paragraph{The Thinning-Out Step.}
In the thinning-out step
we create an auxiliary array, $C$, of size
$r/t_i$, and we perform two $A$-to-$C$ thinning passes,
which will
fail to map any remaining distinguished item in $A$ into $C$
with probability at most $(1/t_i^3)^2=1/t_i^6$, by our induction hypothesis.
Given this mapping, we then perform $t_i$ $C$-to-$D$ thinning passes,
so that the probability that any distinguished item in
$C$ fails to get mapped to $D$ is at most $1/2^{2t_i}$.
Since there may be some items originally in $A$ that are now stuck in
$C$, we grow $A$ by concatenating the current copy of $A$ with $C$.
This grows the size of $A$ to be $n+\sum_{j=1}^i r/t_j$, which is 
less than $n+r/2$ even
accounting for all the times we grew $A$ in previous phases.
In addition,
note that the total time to perform all these passes is $O(n+r)$.

At this point in phase $i$, for each distinguished
element that was in $A$ at the beginning of the phase,
we can bound the event that this element failed to get mapped 
into $D$ with an independent
indicator random variable that is $1$ with probability $1/t_i^4$,
since $t_i^4<2^{2t_i}$.

\paragraph{The Region-Compaction Step.}
In this step of phase $i$, we divide $A$ into contiguous regions, $A_1$, $A_2$,
and so on, each spanning $2^{4t_i}$ cells of $A$.
Let us say that such a region, $A_j$, is \emph{over-crowded} if it
contains more than $r_i=2^{4t_i}/t_i^2$ distinguished items.
We note the following.

\begin{lemma}
\label{lem:region}
The probability that any given region, $A_j$, is over-crowded is less than
$2^{-4^{t_i}}$.
\end{lemma}
\begin{proof}
Given the way that the thinning pass is performed, we can bound the
probability of failure of items being copied by independent indicator
random variables that are $1$ with probability $1/t_i^4$.
We can bound expected value, $\mu$, of the sum, $X$, of these 
random variables with
\[
\mu' = \frac{2^{4t_i}}{t_i^4} .
\]
We can use the Chernoff bound of Lemma~\ref{lem:chernoff}
and the fact that $t_i^2\ge 2^4$,
then, to bound
\begin{eqnarray*}
\Pr\left(X > \frac{2^{4t_i}}{t_i^2}\right) &=& 
\Pr\left(X > t_i^2 \mu'\right) \\
&\le& 2^{-t_i^2\mu'} \\
&=& 2^{-2^{4t_i}/t_i^2} \\ 
&<& 2^{-2^{2t_i}} \\ 
&=& 2^{-4^{t_i}},
\end{eqnarray*}
since $2^{2t_i}>t_i^2$, for $t_i\ge2^2$.
\end{proof}
(for Lemma~\ref{lem:region})

So, for each region, $A_j$, we apply Theorem~\ref{thm:tight-sparse} 
to compact the 
distinguished elements in $A_j$ in an oblivious fashion, assuming
$A_j$ is not over-crowded.
Note that this method runs in $O(|A_j|+r_i\log^2 r_i)$ time,
which is $O(|A_j|)$.
Moreover, assuming that $A_j$ is not over-crowded,
this procedure fails for any such $A_j$ (independently) 
with probability at most
\[
\frac{1}{r_i^3} = \frac{t_i^6}{2^{12t_i}} ,
\]
which is at most
\[
\frac{1}{2^{9t_i}} ,
\]
since $t_i^6\le 2^{3t_i}$, for $t_i\ge 2^2$.

Next, for each subarray, $A_j'$,
consisting of the first $r$ cells in a region $A_j$, we perform $t_i^2$
$A_j'$-to-$D$ thinning passes.
Thus, we can model the event that a distinguished
member of such an $A_j'$ fails to be inserted into $D$ 
with an independent indicator random variable that is $1$ with probability at
most
\[
\frac{1}{2^{2t_i^2}}.
\]
In addition, we have the following.

\begin{lemma}
\label{lem:compress}
The total number of distinguished elements that belong to a subarray $A_j'$ 
yet fail
to be mapped to $D$ is more than $r/(2\cdot 2^{4t_i})$ with probability 
at most $1/(2\cdot n^d)$, for any given constant $d\ge 1$.
\end{lemma}
\begin{proof}
There are at most $r$ distinguished elements remaining in $A$ at the
beginning of phase $i$ that could possibly belong to a subarray $A_j$.
The probability that these fail to be mapped to $D$ can be
characterized with a sum, $X$, of independent indicator random variables, each
which is $1$ with probability at most $2^{-2t_i^2}$.
Thus, the expected value of $X$ can be bound by
\[
\mu = \frac{r}{2^{2t_i^2}},
\]
and we can apply the Chernoff bound of Lemma~\ref{lem:chernoff} as follows:
\begin{eqnarray*}
\Pr\left( X > \frac{r}{2\cdot 2^{4t_i}}\right) 
  &\le& \Pr\left(X > \frac{2^{2t_i^2}}{2\cdot 2^{4t_i}} \mu\right) \\
  &\le& 2^{-r/(2\cdot 2^{4t_i})}.
\end{eqnarray*}
Moreover,
since $r/t_i^4 < n/\log^2 n$, we have that $t_i < 0.25\log^{1/2} n$.
Thus,
\begin{eqnarray*}
\frac{r}{2\cdot 2^{4t_i}} &\ge& 
     \frac{n}{\log^2 n \cdot 2 \cdot 2^{\sqrt{\log n}}} \\
     &\ge& d\log n + 1,
\end{eqnarray*}
for $n$ larger than a sufficiently large constant, $n_0$.
\end{proof}
(for Lemma~\ref{lem:compress})

Of course, a distinguished member of $A$ may fail to be mapped into an
$A_j'$ subarray,
either because it belongs to an over-crowded region or it belongs to a
region that fails to be compacted.
Fortunately, we have 
the following.

\begin{lemma}
\label{lem:crowded}
The number of distinguished elements in $A$ that fail to be mapped into $D$
in phase $i$, either because they belong to an over-crowded region
or because their region fails to be compacted,
is more than $r/(2\cdot 2^{4t_i})$ with probability at most $1/(2n^d)$.
\end{lemma}
\begin{proof}
Recall that the probability that a region $A_j$ is over-crowded is at most
$2^{-4^{t_i}}$ and the probability that a non-over-crowded region $A_j$
fails to be mapped to a compressed region, $A_j'$, is at most
$2^{-9t_i}$. 
Let us consider the latter case first.

Let $X$ denote
the number of distinguished elements that belong to
non-over-crowded regions that fail to be compacted in phase $i$.
Note that the number of non-empty regions is at most $r/t_i^4$ and the number
of distinguished items per such region is clearly no more than $2^{4t_i}$.
Thus, if we let $Y$ be the number of non-empty non-over-crowded 
regions that fail to be compressed, then
\[
X \le 2^{4t_i} Y .
\]
Moreover, since $Y$ is the sum of at most $r/t_i^4$ independent indicator
random variables, and 
\[
E(Y) \le \mu = r/(t_i^4 2^{9t_i}),
\]
we can use the Chernoff
bound from Lemma~\ref{lem:chernoff} as follows:
\begin{eqnarray*}
\Pr\left( X > \frac{r}{4\cdot 2^{4t_i}}\right) 
   &\le& \Pr\left( Y > \frac{r}{2^{8t_i+2}}\right) \\
   &\le& \Pr\left( Y > 2^{t_i-2}t_i^4 \mu\right) \\
   &\le& 2^{-r/2^{8t_i}} \\
   &\le& 2^{-n/(\log^2 n \cdot 2^{2\sqrt{\log n}})},
\end{eqnarray*}
since $t_i\le 0.25\log^{1/2} n$ and $t_i^4\ge 2^8$.
Note that, for $n$ greater than a sufficiently large constant, $n_0$,
\[
   \frac{n}{\log^2 n \cdot 2^{2\sqrt{\log n}}} \ge d\log n + 2,
\]
for any fixed constant $d\ge 1$.
Thus, the number of distinguished elements that belong to
non-over-crowded regions that fail to be compacted in phase $i$ 
is more than $r/(4\cdot 2^{4t_i})$ with probability at most $1/(4n^d)$.
By a similar argument,
the number of distinguished items that belong to over-crowded regions
is more than $r/(4\cdot 2^{4t_i})$ with probability at most $1/(4n^d)$.
In fact, this case is easier,
since $t_i\ge 4$ implies $2^{-4^{t_i}}<2^{9t_i}$.
Therefore,
the number of distinguished elements that fail for either of these reasons is
more than $r/(2\cdot 2^{4t_i})$ with probability at most $1/(2n^d)$.
\end{proof}
(for Lemma~\ref{lem:crowded})

Thus, combining this lemma with Lemma~\ref{lem:compress}, we have that
the number of
distinguished elements in $A$ that are not mapped (data-obliviously)
into $D$ in phase $i$ is
at most
\[
\frac{r}{2^{4t_i}} = \frac{r}{t_{i+1}^4} ,
\]
with probability at least $1-1/n^d$, 
which gives us the induction invariant for the next phase.
This completes the proof (noting that
we can apply the above algorithm with $d'=d+1$, since there are $O(\log^* n)$
phases, each of which runs in $O(n)$ time and fails with 
probability at most $1/n^{d'}$).
\end{proof}
(for Theorem~\ref{thm:logstar})
\fi

\end{appendix}
\end{document}